\documentclass[11pt]{article}
\usepackage{amsmath, amssymb, amstext, amsthm,epsfig}
\usepackage{graphicx}
\usepackage[cp1251]{inputenc}
\usepackage{amsmath, amssymb}
\usepackage{mathtools}
\newcommand{\poly}{\text{poly}}
\newtheorem{hypothesis}{Hypothesis}

\newcommand{\cnd}{\mskip 1mu|\mskip 1mu}
\DeclareMathOperator{\C}{\mathrm{C}}
\DeclareMathOperator{\CD}{\mathrm{CD}}

\newcommand{\RL}{{\mathrm{RL}}}
\newcommand{\SC}{{\mathrm{SC}}}

\theoremstyle{plain}
\newtheorem{theorem}{Theorem}
\newtheorem*{lemma}{Lemma}

\theoremstyle{remark}
\newtheorem{remark}{Remark}

\newtheorem{example}{Example}

\begin{document}
	%
	%\frontmatter          % for the preliminaries
	%
	\pagestyle{headings}  % switches on printing of running heads
	%\addtocmark{Hamiltonian Mechanics} % additional mark in the TOC
	%
	
	%\mainmatter              % start of the contributions
	%
	\title{On Algorithmic Statistics for space-bounded algorithms}
	
	\author{Alexey Milovanov \\
		National Research University Higher School of Economics, \\
		Moscow Institute of Physics and Technology, \\
		Moscow State University, \\
		almas239@gmail.com}  
	
	\maketitle              % typeset the title of the contribution
	
	\begin{abstract}
		Algorithmic statistics studies explanations of observed data that are good in
		the algorithmic sense:  an explanation should be simple i.e. should have small Kolmogorov complexity and capture all the
		algorithmically discoverable regularities in the data. However this idea can not be used in practice because Kolmogorov complexity is not computable.
		
		In this paper we develop algorithmic statistics using space-bounded Kolmogorov complexity. We prove an analogue of one of the main result of `classic' algorithmic statistics (about the connection between optimality and randomness deficiences). The main tool of our proof is the Nisan-Wigderson generator.
	
	\end{abstract}
	\section{Introduction}
	In this section we give an introduction to algorithmic statistics and present our results.
	
	We consider strings over the binary alphabet $\{0,1\}$. We use $|x|$ to denote the length of a string $x$. All of the logarithms are base $2$. Denote the conditional Kolmogorov complexity\footnote{The definition and basic properties of Kolmogorov complexity can be
		found in the textbooks \cite{lv,suv}, for a short survey see \cite{shen15}.} of $x$ given $y$ by $\C (x \cnd y)$. 
	\subsection{Introduction to  Algorithmic Statistics}
	Let $x$ be some observation data encoded as a binary string, we need to find a suitable explanation for it. An  explanation ($=$model) is a finite set containing $x$. More specifically we want to find a simple model $A$ such that $x$ is a typical element in $A$. How to formalize that $A$ is `simple' and $x$ is a `typical element' in $A$? In classical algorithmic statistics a set $A$ is called simple if it has small Kolmogorov complexity $\C (A)$\footnote{Kolmogorov complexity of $A$ 
		is defined as follows. We fix any computable bijection $A \mapsto [A]$ 
		from the family of finite sets
		to the set of binary strings, called \emph{encoding}.
		Then we define $\C (A)$ as the complexity $\C ([A])$ of the code $[A]$ of $A$.}. To measure typicality of $x$ in $A$ one can use the \emph{randomness deficiency} of $x$ as an element of $A$:
	$$ d(x \cnd A) := \log|A| - \C (x \cnd A). $$
	The randomness deficiency is always non-negative 
	with $O(\log |x|)$ accuracy, as 
	we can find $x$ from $A$ and the index of $x$ in $A$.  
	For most elements $x$ in any set $A$ the randomness deficiency 
	of $x$ in $A$ is negligible.
	More specifically, the fraction of $x$ in $A$
	with randomness deficiency greater than $\beta$ is less than 
	$2^{-\beta}$. 
	
	There is another quantity  measuring the quality of $A$ as an explanation of $x$: the \emph{optimality deficiency}:
	$$\delta(x,A):= \C (A) + \log|A| - \C (x). $$ 
	It is also non-negative with logarithmic accuracy (by the same reason). This value represents the following idea: a good explanation (a set) should not only be simple but also should be small. 
	
	One can ask: why as explanations we consider only sets---not general probability distributions? This is because for every string $x$ and for every distribution $P$ there exists a set $A \ni x$ explaining $x$ that is not worse than $P$ in the sense of deficiencies defined above\footnote{The randomness deficiency of a string $x$ with respect to a distribution $P$ is defined as $d(x \cnd P):= -\log P(x) - \C (x \cnd P)$, the optimality deficiency is defined as $\delta(x,P):= \C (P) - \log P(x) - \C (x)$. }.
	\begin{theorem}[\cite{vv}]
		\label{sets_and_distributions0}
		For every string $x$  and for every distribution $P$ there exists a set $A \ni x$ such that $\C (A \cnd P) \le O(\log |x|)$ and 
		$\frac{1}{|A|} \ge \frac{1}{2}P(x)$.
	\end{theorem}
	Kolmogorov called a string $x$ \emph{stochastic} if there exists a set $A \ni x$ such that $\C (A) \approx 0$ and $d(x \cnd A)\approx 0$. The last equality means that $\log|A| \approx \C (x \cnd A)$ hence $\log|A| \approx \C (x)$ because $\C (A) \approx 0$. So, $\delta(x,A)$ is also small.
	
	For example, an incompressible string of length $n$ (i.e. a string whose complexity is close to $n$) is stochastic---the corresponding set is $\{0,1\}^n$. Non-stochastic objects also exist, however this fact is more complicated---see \cite{shen,vv}.
	\subsection{Space-bounded Algorithmic Statistics}
	As mentioned by Kolmogorov in \cite{kolm65}, the notion of Kolmogorov complexity $\C(x)$ has the following minor point.  It ignores time and space needed to produce $x$ from its short description. This minor point can be fixed by introducing space or time bounded Kolmogorov complexity (see, for example, \cite{bfl} or \cite{sipser}). In this paper we consider algorithms whose space (not time) is bounded by a polynomial of the length of a string.
	
	The distinguishing complexity of a string $x$ with space bound $m$ is defined as the minimal length of a program $p$ such that
	\begin{itemize}
		\item $p(y)=1$ if $y = x$;
		\item $p(y) = 0$ if $y \not = x$;
		\item $p$ uses at most $m$ bits of memory on every input.
	\end{itemize}
	We denote this value by $\CD^m(x)$. If for some $x$ and $m$ such a program $p$ does not exist then $\CD^m(x):= \infty$.
	We say that $p$ \emph{distinguishes} $x$ (from other strings) if $p$ satisfies the first and the second requirements of the definition. 
	
	In this definition  $p(y)$ denotes $V(p,y)$ for a universal Turing machine $V$. A Turing machine is called universal if
	for every machine $U$ and for every $q$ there exists $p$ such that $V(p,y) = U(q,y)$ for every $y$, $|p| < |q| + O(1)$  and $V$ uses space at most $O(m)$  if $U$ uses space $m$  on input $(q,y)$. Here the constant in $O(m)$ depends on $V$ and $U$ but does not depend on $q$ \footnote{Such an universal machine does exist -- see \cite{lv}.}.
	
	Now we extend this notion to arbitrary finite sets. The distinguishing complexity of a set $A$ with space bound $m$ is defined as the minimal length of a program $p$ such that
	\begin{itemize}
		\item $p(y)=1$ if $y \in  A$;
		\item $p(y) = 0$ if $y \notin  A$;
		\item $p$ uses space $m$ on every input.
	\end{itemize}
	Denote this value as $\CD ^m(A)$.
	
	The value $\CD^a(x \cnd A)$ is defined as the minimal length of a program that distinguishes $x$ by using space at most $m$  and uses $A$ as an oracle. The value $\CD^a(B \cnd A)$ for an arbitrary finite set $B$ is defined the same way.
	
	How to define typicality of a string $x$ in a set $A$? Consider the following 
	resource-bounded versions of randomness and optimal deficiencies:
	$$d^a(x \cnd A):=\log|A| - \CD^a(x \cnd A),$$
	$$\delta^{b,d}(x, A):= \CD^b(A) + \log|A| -  \CD^d(x) .$$
	One can show that these values are 
	non-negative (with logarithmic accuracy) provided $a \ge p(|x|)$ and $d \ge p(|x| + b)$ for a large enough polynomial  $p$.
	
	We say that a set $A$ is a good explanation for a string $x$ (that belongs to $A$) if $\CD^r(A) \approx 0$ (with $O(\log |x|)$ accuracy) and $\log|A| \approx \CD^m(x)$. Here $r$ and $m$ are some small numbers. For such $A$ the values 
	$d^m(x \cnd A)$ and $\delta^{r,m}(x, A)$ are small.
	
	It turns out that every string has a good explanation. Indeed,
	let $x$ be a string such that $\CD^m(x) = k$. Define a set $A \ni x$ as  $\{y \mid \CD^m(y) \le k\}$. The log-size of this set is equal to $k$ up to a non-negative constant and hence $\log|A| = \CD^m(x)$. Note that $A$ can be distinguished by a program of length $O(\log (k + m))$ that uses $\poly(m)$ space.
	
	So, for space-bounded algorithms all strings have  
	good explanations (in other words, they are stochastic).
	
	\subsection{Distributions and Sets}
	Recall that in the classical algorithmic statistics for every distribution $P$ and every $x$ there is a finite set $A \ni x$ that is not worse than $P$ as an explanation for $x$. It turns out that  this the case also for space-bounded algorithmic statistics (otherwise we could not restrict ourselves to finite sets). 
	% We want to develop a non-trivial theory (where not all objects are stochastic). For this one can try to use general distributions as explanations for a string (not only uniform distributions i.e. sets). If general distributions provide significantly better explanations than sets then we have a chance for a substantial theory.
	
	%(Note that non-negativity of the space-bounded optimality deficiency of $x$ with respect to a probability distribution $P$ (that is defined by a natural way) is not obvious.)
	
	% Else if for resource-bounded algorithmic statistics there exists an analogue of Theorem \ref{sets_and_distributions0} then we can consider only sets and so, all strings will be stochastic  in senses of randomness and optimal deficiencies (as it is shown in the previous subsection). 
	
	%It turns out that such an analogue exists.
	Before we formulate this result we give a definition of  the complexity of a probability distribution $P$ with space bound $m$ that is denoted by $\C^m(P)$. This value is defined as the minimal length of a program $p$ without input and with the following two properties. First, for every $x$ the probability of the event [$x$  output by $p$] is equal to $P(x)$. Second, $p$ uses space at most $m$ (always). If such a program does not exist then $\C^m (P) :=\infty$.
	\begin{theorem}
		\label{sets_and_distributions_for_memory1}
		There exist a polynomial $r$ and a constant $c$ such that for every string $x$, for every distribution $P$ and for every $m$ there exists a set $A \ni x$ such that $\CD ^{r(m + n)}(A) \le \C ^{m}(P) + c\log (n + m)$ and $\frac{1}{|A|} \ge P(x)2^{-c \log n}$. Here $n$ is length of $x$.  
	\end{theorem} 
	The main tool of the proof of Theorem \ref{sets_and_distributions_for_memory1} is the theorem of Nisan ``$\RL \subseteq \SC$'', more precisely its  generalization---Theorem 1.2 in  \cite{nisan}.
	\subsection{Descriptions of Restricted Type}
	So far we considered arbitrary finite sets 
	(or more general distributions) as models (statistical hypotheses). We have seen that for such class of hypotheses the theory becomes trivial. However, in practice we usually have some a priori information about the data. We know that the data was obtained by sampling with respect to an unknown probability distribution from a known family of distributions. For simplicity we will consider only uniform distributions i.e. a family of finite sets $\mathcal{A}$. 
	
	For example, we can consider the family of all Hamming balls as $\mathcal{A}$. (That means we know a priory that our string was obtain by flipping certain number of bits in an unknown string.) 
	Or we may consider the family that consists of all `cylinders': for every
	$n$ and for every string $u$ of length at most $n$ we consider the set of all
	$n$-bit strings that have prefix $u$.
	It turns out that for the second family
	there exists a string that has no  good explanations in this family: the concatenation of an incompressible string (i.e. a string whose Kolmogorov complexity is close to its length) and all zero string  of the same length. (We omit the rigorous formulation and the proof.)
	
	Restricting the class of allowed hypotheses was initiated in \cite{VerVit}.
	It turns out that there exists a direct connection between randomness and optimality deficiencies in the case when   a family is  enumerable. 
	\begin{theorem}[\cite{VerVit}]
		\label{without_bound}
		Let $\mathcal{A}$ be an enumerable family of sets. Assume that every set from $\mathcal{A}$ consists of strings of the same length.
		Let $x$ be a string of length $n$ contained in $A \in \mathcal{A}$. Then:
		\medskip
		
		\textup{(a)} $d(x \cnd A) \le \delta(x,A) + O(\log (\C(A) + n ))$.
		
		\medskip
		
		\textup{(b)}
		There exists $B \in \mathcal{A}$ containing $x$ such that:
		
		\medskip
		
		$\delta(x, B) \le d(x \cnd A) + O(\log (\C(A) + n ))$.
	\end{theorem}
	In our paper we will consider families  with the following properties:
	\begin{itemize}
		\item Every set from $\mathcal{A}$ consists of strings of the same length. The family of all subsets of $\{0,1\}^n$ that belong to  $\mathcal{A}$ is denoted by $\mathcal{A}_n$.
		\item There exists a polynomial $p$ such that $|\mathcal{A}_n| \le 2^{p(n)}$ for every $n$. 
		\item There exists an algorithm enumerating all sets from $\mathcal{A}_n$ in space $\poly(n)$.
	\end{itemize}
	The last requirement means the following. There exists an indexing of $\mathcal{A}_n$ and a Turing machine $M$  that for a pair of integers $(n; i)$ and a string $x$ in the input outputs $1$ if $x$ belongs to $i$-th set of $\mathcal{A}_n$ and $0$ otherwise. On every such input $M$ uses at most $\poly(n)$ space. 
	
	Any family of finite sets of strings
	that satisfies these three conditions is called \emph{acceptable}. For example, the family of all Hamming balls is  acceptable.
	Our main result is the following analogue of Theorem \ref{without_bound}.
	\begin{theorem}
		\label{def}
		\textup{(a)}
		There exist a polynomial $p$ and a constant $c$ such that for every set $A \ni x$ and for every $m$ the following inequality holds
		$$d^m(x \cnd A) \le \delta^{m,p}(x,A) + c\log ( C^m(A)).$$
		Here $p=p(m+n)$ and  $n$ is the length of $x$.
		
		\textup{(b)}
		For every acceptable family of sets $\mathcal{A}$ there exists a polynomial $p$ such that the following property holds.  For every $A \in \mathcal{A}$, for every $x \in A$ and for every integer $m$ there exists a set $B \ni x$ from $\mathcal{A}$ such that
		\begin{itemize}
			\item
			$
			\begin{array}{l}
			\log|B| \le \log|A| + 1;
			\end{array}
			$
			\item
			$
			\begin{array}{l}
			\CD ^s(B) \le \CD ^m(A) - \CD ^s(A \cnd x) + O(\log (n + m)).
			\end{array}
			$
		\end{itemize}
		Here $s = p(m + n)$ and $n$ is the length of $x$.
	\end{theorem}
	
	A skeptical reader would say that an analogue of Theorem~\ref{without_bound} \textup{(b)} should has the following form (and we completely agree with him/her).
	\begin{hypothesis}
		\label{hyp}
		There exist a polynomial $p$ and a constant $c$ such that for every set
		$A \ni x$ from $\mathcal{A}$ and for every $m$ there exists a set $B \in \mathcal{A}$ such that
		$$\delta^{p,m}(x, B) \le d^p(x \cnd A) + c\log (n + m).$$
		Here $p = p(m + n)$,  $n$ is the length of $x$ and $\mathcal{A}$ is an acceptable family of sets.
	\end{hypothesis}
	We argue in Subsection~\ref{sketch} why
	Theorem~\ref{def}~\textup{(b)} is close to Hypothesis \ref{hyp}. 
	\section{Proof of Theorem \ref{def}}
	\begin{proof}[of Theorem \ref{def}\textup{(a)}]
		The inequality we have to prove means the following
		$$\CD ^p(x) \le \CD ^m(x \cnd A) + \CD ^m(A) + c \log( \CD ^m(A) + n)$$
		(by the definitions of optimality and randomness deficiencies).
		
		Consider a program $p$ of length  $\CD ^m(x \cnd A)$
		that distinguishes $x$ and uses $A$ as an oracle.
		We need to construct a program that also distinguishes $x$ but does not use any oracle. For this add to $p$ a procedure distinguishing $A$. There exists such a procedure of length $\CD ^m(A)$. So, we get a program of the  length that we want (additional $ O(\log( \CD ^m(A)))$ bits are used for pair coding) that uses  $\poly(m)$ space.
	\end{proof}
	
	So, for every $x$ and $A \ni x$ the randomness deficiency is not greater than the optimal deficiency. The following example shows that the difference can be large. 
	\begin{example}
		Consider an incompressible string $x$ of  length $n$, so $\C (x) = n$ (this equality as well as further ones holds with logarithmic precision). Let $y$ be $n$-bit string that is also incompressible and independent of  $x$, i.e. $\C (y \cnd x) = n$.
		By symmetry of information (see \cite{suv,lv}) we get $\C (x \cnd y) = n$.  
		
		Define $A:=\{0,1\}^n \setminus \{y\}$.
		The randomness deficiency of $x$ in $A$ (without  resource restrictions) is equal to $0$. Hence, this is true for any resource restrictions ($\C (x \cnd A)$  is not greater than $\CD ^m(x \cnd A) $ for every $m$).   
		Hence, for any $m$ we have $d^m(x \cnd A) = 0$. On the other hand  $\delta^p_m(x, A) =n$ for all $p$ and large enough $m$. Indeed, take $m = \poly(n)$ such that $\CD ^m(x) = n$. Since $\C(A) =n$ we have $\CD ^q(A) =n$ for every~$q$.
	\end{example}
	So, we can not just let $A=B$ in Hypothesis~\ref{hyp}. In some cases we have to `improve' $A$ (in the example above we can take  $\{0,1\}^n$ as an improved set). 
	\subsection{Sketch of proof of Theorem \ref{without_bound}\textup{(b)}}
	\label{sketch}
	The proof of Theorem~\ref{def}~\textup{(b)} is similar to the proof of Theorem~\ref{without_bound}~\textup{(b)}.  Therefore we  present the sketch of the proof of Theorem~\ref{without_bound}~\textup{(b)}. 
	
	Theorem~\ref{without_bound} states that  there exists a set $B \in \mathcal{A}$ containing $x$ such that $\delta(x \cnd B) \le d(x, A)$.
	(Here and later we omit terms of logarithmic order.)
	First we derive it from the following statement.
	
	(1) There exists a set $B \in \mathcal{A}$ containing $x$ such that
	
	$|B| 
	\le 2 \cdot |A|$ and $\C (B) \le \C (A) - \C (A \cnd x)$.
	
	For such $B$ the 
	$\delta(x \cnd B) \le d(x, A)$ easily follows from the inequality
	$\C (A) - \C (A \cnd x) - \C (x) \le - \C (x \cnd A) $. The latter inequality holds by symmetry of information. 
	
	To prove (1) note that
	
	(2) there exist at least $2^{\C (A \cnd x)}$ sets in  $\mathcal{A}$ containing $x$ whose complexity and size are at most $\C (A)$ and 
	$2 \cdot |A|$, respectively.
	
	Indeed, knowing $x$ we can enumerate all sets from $\mathcal{A}$ containing $x$ whose parameters (complexity and size) are not worse than the parameters of $A$. Since we can describe $A$ by its ordinal number in this enumeration we conclude that the length of this number is at least $\C (A \cnd x)$ (with logarithmic precision).
	
	Now (1) follows from the following statement.
	
	(3) Assume that $\mathcal{A}$ contains at least $2^k$ sets of complexity at most $i$ and size at most $2^j$ containing $x$. Then one of them has complexity at most $i-k$.
	
	(We will apply it to $ i=\C (A)$, $j=\lceil \log|A| \rceil$ and $k=\C (A \cnd x)$.)
	
	So, Theorem \ref{def} (\textup{b}) is an analogue of (2). Despite there is an analogue of symmetry of information for space-bounded algorithms (see \cite{longpre} and Appendix) Hypothesis \ref{hyp} does not follow  Theorem \ref{def} (\textup{b}) directly. (There is some problem with quantifiers.)
	
	Proof of (3) is the main part of the proof of Theorem~\ref{without_bound}, the same thing holds for Theorem~\ref{def}.
	
	In the next subsection we derive Theorem \ref{def} (\textup{b}) from Lemma \ref{nw} (this is an analogue of the third statement). In the proof of Lemma~\ref{nw} we use the Nisan-Wigderson generator.
	\subsection{Main lemma}
	We will derive  Theorem \ref{def} (\textup{b}) from the following 
	\begin{lemma}
		\label{nw}
		For every acceptable family of sets $\mathcal{A}$ there exist a polynomial $p$ and a constant $c$ such that the following statement holds for every $j$.
		
		Assume that a string $x$ of length $n$ belongs to $2^k$ sets from $\mathcal{A}_n$. Assume also that every of these sets has cardinality at most $2^j$ and space-bounded by $m$ complexity at most $i$.  Then one of this set is space-bounded by $M$ complexity at most  $i - k + c \log (n + m)$. Here $M = m + p(n)$.
	\end{lemma}
	\begin{proof}[ Theorem \ref{def} (\textup{b}) from Lemma~\ref{nw}]
		
		Denote by $\mathcal{A}'$  the family of all sets
		in $\mathcal{A}_n$ containing $x$ whose parameters are not worse than those of $A$.
		$$\mathcal{A}':= \{ A' \in \mathcal{A}_n \mid 
		x \in A', \CD ^m(A) \le \CD ^m(A'), \log |A'| \le \lfloor \log |A|\rfloor  \}. $$
		Let $k = \log \mathcal{A}'$.
		
		We will describe $A$ in $k + O(\log(n + m))$ bits when $x$ is known.
		The sets in $\mathcal{A}'$ (more specifically, their programs) can be enumerated if $n,m$ and  $\log|A|$ are known. This enumeration can be done in space $\poly(m + n)$. We can describe $A$ by its ordinal number of this enumeration, so
		$$ \CD ^s(A \cnd x) \le k + O(\log (n + m) ).$$
		Here $s = \poly(m + n)$.
		
		Theorem \ref{def} (\textup{b}) follows from Lemma~\ref{nw} for
		$i = \CD ^m(A)$ and $j = \lfloor \log |A|\rfloor$.
	\end{proof}
	
	\subsection{Nisan-Wigderson generator. Proof of the main lemma }
	Define $$\mathcal{A}_{n,m}^{i,j}:= \{ A' \in \mathcal{A}_n \mid \CD ^m(A') \le i, \log|A'| \le j \}$$ for an acceptable family of sets $\mathcal{A}$.
	
	Define a probability distribution $\mathcal{B}$ as follows. Every set from $\mathcal{A}_{n,m}^{i,j}$ belongs to $\mathcal{B}$ with probability  $2^{-k}(n +2) \ln 2$ independently. 
	
	We claim that  $\mathcal{B}$ satisfies the following two properties with high probability.
	
	\textup{(1)} The cardinality of $\mathcal{B}$ is at most  $2^{i - k + 2} \cdot (n+k)^2 \ln 2$.
	
	\textup{(2)}
	If a string of length $n$ is  contained in at least $2^k$ sets from $\mathcal{A}_{n,m}^{i,j}$ then one of these sets belongs to $\mathcal{B}$.
	
	\begin{lemma}
		\label{prob}
		The family 
		$\mathcal{B}$ satisfies the properties \textup{(1)} and \textup{(2)} with probability at least~$\frac{1}{2}$.  
	\end{lemma}
	
	\begin{proof}
		Show that $\mathcal{B}$ satisfies every of these two properties with probability at least $\frac{3}{4}$.
		
		For \textup{(1)} it follows from Markov's inequality: the cardinality of $\mathcal{B}$ exceeds the expectation by a factor of $4$ with probability less than $\frac{1}{4}$. (Of course we can get a rather more stronger estimation.)
		
		To prove it for \textup{(2)} consider a string of length $n$ that belongs to at least $2^k$ sets from $\mathcal{A}_{n,m}^{i,j}$. The probability of the event [every of these $2^k$ sets does not belong to $\mathcal{B}$] is at most 
		$$(1 - 2^{-k}(n+2)\ln 2)^{2^k}\le 2^{-n-2} \text{ (since }1 - x \le e^{-x}).$$
		The probability of the sum of such events for all strings of  length $n$ is at most $2^n 2^{-n-2} = \frac{1}{4}$.
	\end{proof}
	
	Using  Lemma \ref{prob} we can prove existence of a required set whose  \emph{unbounded} complexity is at most $i - k + O(\log (n + m))$. Indeed, by  Lemma~\ref{prob} \emph{there exists} a subfamily that satisfies the properties \textup{(1)} and \textup{(2)}. The lexicographically first such family has  small complexity---we need only know $i$, $k$, $n$ and $m$ to describe it. Note, that $k$ and $i$ are bounded by $\poly(n)$: since $\mathcal{A}$ is acceptable $\log|\mathcal{A}_n|=\poly(n)$ and hence $k$ is not greater than $\poly(n)$. We can enumerate all sets from $\mathcal{A}_n$, so space-bounded complexity of every element of  
	$\mathcal{A}_n$ (in particular, $i$) is bounded by polynomial in $n$.
	Now we can describe a required set as the ordinal number of an enumeration of this subfamily.
	
	However, this method is not suitable for the polynomial space-bounded complexity: the brute-force search for the finding a suitable subfamily uses too much space (exponential). To reduce it we will use the Nisan-Wigderson generator. The same idea was used in \cite{musatov}.
	\begin{theorem}[\cite{nis,nw}]
		\label{tnw}
		For every constant $d$ and for every positive polynomial $q(m)$ there exists a sequence of functions $G_m: \{0,1\}^f \to \{0,1\}^m$ where $f = O(
		\log^{2d+6} m)$ such that:
		\begin{itemize}
			\item
			Function $G_m$ is computable in space $\poly(f)$;
			\item
			For every family of circuits $C_n$ of size $q(v)$ and depth $d$ and for large enough $n$ it holds that:
			$$|\Pr_x[C_m(G_m(x)) = 1] - \Pr_y[C_m(y) = 1]| < \frac{1}{m},$$ 
			where $x$ is distributed uniformly in $\{0,1\}^f$, and  $y$ is distributed uniformly in $\{0,1\}^m$.
		\end{itemize} 
	\end{theorem} 
	We will use this theorem for $m = 2^{i+n}$.
	Then $f$ is a polynomial in $i+n$ (if $d$ is a constant), hence $f = \poly(n)$. Every element whose complexity is at most $i$ corresponds to a string of length $i$ in the natural way. So, we can assign 
	subfamilies of 
	$\mathcal{A}_{n,m}^{i,j}$ to  strings of length $m$. 
	
	Assume that there exists a circuit of size $2^{O(n)}$ and constant depth  that inputs 
	a subfamily of $\mathcal{A}_n^{i,j}$ and outputs {\bf $1$} if this subfamily satisfies properties \textup{(1)} and \textup{(2)} from Lemma \ref{prob}, and {\bf $0$} otherwise. First we prove Lemma~\ref{nw} using this assumption.
	
	Compute $G_m(y)$ for all strings $y$ of length  
	$f$ until we find a suitable one, i.e. whose image satisfies our two properties. 
	Such a string exists by  Lemma~\ref{prob}, Theorem~\ref{tnw} and our assumption. Note that we can find the lexicographically first  suitable string by using  space $m + \poly(n)$, so   \emph{bounded by space}
	$m + \poly(n)$  the complexity of this string is equal to $O(\log (n + m))$.  
	
	So, if we can construct a constant depth circuit of the needed size that verifies properties \textup{(1)} and \textup{(2)} then we are happy. Unfortunately we do not know how to construct such a circuit verifying the first property (there exist problems with a computation of threshold functions by constant-depth circuits---see \cite{fss}). However, we know the following result.
	\begin{theorem}[\cite{ajtai}]
		\label{porog}
		For every $t$ there exists a circuit of constant depth and $\poly(t)$ size that inputs binary strings of length $t$ and outputs $1$ if an input has at most $\log^2 t$ ones and $0$ otherwise. 
	\end{theorem}
	To use this theorem we make a little change of the first property. Divide $\mathcal{A}_n^{i,j}$ into $2^{i-k}$ parts of size $2^k$. The corrected property is the following.
	
	\textup{$(1)^*$} The family of sets $\mathcal{B}$ contains at most $(n + k)^2$ sets from each of these parts.
	
	\begin{lemma}
		\label{prob1}
		The family of sets $\mathcal{B}$ satisfies properties \textup{$(1)^*$} and \textup{(2)} with probability at least $\frac{1}{3}$.
	\end{lemma}
	The proof of this lemma is not difficult but uses cumbersome formulas. We present the proof of Lemma \ref{prob1} in Appendix.
	\begin{proof}[of Lemma \ref{nw}]
		It is clear that property \textup{$(1)^*$} implies property \textup{(1)}. Hence by using Lemma \ref{nw} and the discussion above, it is  enough to show that properties \textup{$(1)^*$} and \textup{(2)} can be verified by constant depth circuits of size
		$2^{O(i + n)}$.
		
		Such a circuit  exists for property \textup{$(1)^*$} by Theorem \ref{porog}.
		
		The second property can be verified by the following $2$-depth circuit. For every string of length $n$ containing in $2^k$ sets from $\mathcal{A}_n^{i,j}$ there exists a corresponding disjunct. All of these disjuncts go to a conjunction gate.

	\end{proof}
	\section{Proof of Theorem \ref{sets_and_distributions_for_memory1}}
	Theorem \ref{sets_and_distributions_for_memory1} would have an easy proof if a program that corresponds to a distribution $P$ could use only $\poly(n)$  random bits. Indeed, in such case we can run a program with all possible random bits and so calculate $P(x)$ for every $x$ in polynomial space. Hence, we can describe $A$ as the set of all strings whose the probability of output is at least $2^{-k}$, where $2^{-k} \ge P(x) > 2^{-k-1}$. 
	
	In the general case (when the number of random bits is exponentially large) we will use the following theorem.
	\begin{theorem}[\cite{nisan}] 
		\label{nis}
		Let $f$ be a probabilistic program, that uses at most $r(n)$ space on inputs of length $n$ for some polynomial $r$. Assume that $f$ always outputs $0$ or $1$ (in particular, $f$ never loops). Then there exists a deterministic program $\widehat{f}$ with the following properties:

		\textup{(a)} $\widehat{f}$ uses at most  $r^2(n)$ space on inputs of length $n$;
		
		\textup{(b)} if $Pr[f(x)=1] > \frac{2}{3}$ then $\widehat{f}(x)=1$.  If $Pr[f(x)=1]  < \frac{1}{3}$ then $\widehat{f}(x)=0$;
		
		\textup{(c)} $|\widehat{f}| \le |f| + O(1)$. \footnote{Theorem 1.2 in \cite{nisan} has another formulation: it does not contain any information about $|\widehat{f}|$. However, from the proof of the theorem it follows that a needed program (denote it as $\widehat{f}_1$) is got from $f$ by using an algorithmic transformation. Therefore there exists a program $\widehat{f}$ that works functionally like $\widehat{f}_1$ such that $|\widehat{f}| \le |f| + O(1)$.
			
			Also, Theorem 1.2 does not assume that $\Pr[f(x)]$ can belong to $[\frac{1}{3}; \frac{2}{3}]$. However, this assumption does not used in the proof of Theorem 1.2.} 
	\end{theorem}
	\begin{proof}[of Theorem \ref{sets_and_distributions_for_memory1}]
		If the complexity of distribution $P$ 
		(bounded by space $m$) is equal to infinity then we can take $\{x\}$ as $A$.
		
		Else $P$ can be specified by a program $g$. Consider the integer $k$ such that: $2^{-k+1} \ge P(x) \ge 2^{-k}$. We can assume that 
		$k$ is not greater than $n$---the length of $x$---else we can take $\{0,1\}^n$ as $A$.   
		
		Note, that we can find a good approximation for $P(y)$ running $g$ exponentially times.
		
		More accurately, let us run $g$ for $2^{100k^2}$ times. For every string $y$ denote by $\omega(y)$ the frequency of output of $y$. The following inequality holds by  Hoeffding's inequality  
		$$\Pr [ |w(y) - P(y)| > 2^{-k-10}]< \frac{1}{3}.$$
		Hence by using program $g$ we can construct a program $f$ that uses $\poly(n)$ space (on inputs of length $n$) such that
		
		(1) if $P(y) > 2^{-k-1}$ and $|y| = n$ then $\Pr [f(y)=1] > \frac{2}{3}$;
		
		\medskip
		
		(2) if $P(y) < 2^{-k-2}$ then $\Pr [f(y)=0] > \frac{2}{3}$. 
		
		Now using Theorem \ref{nis} for $f$ we get a program $\widehat{f}$ such that $|\widehat{f}| \le |g| + O(\log n)$.
		By the first property of $f$ we get $\widehat{f}(x)=1$. From the second property it follows that the cardinality of the set $\{ y \mid \widehat{f}(y)=1  \}$ is not greater than $2^{k+2}$. So, this set satisfies the requirements of the theorem.
	\end{proof}
	
	\begin{remark}
		Another proof of Theorem \ref{sets_and_distributions_for_memory1}
		was done by Ricky Demer at Stackexchange -- http://cstheory.stackexchange.com/questions/34896/can-every-distribution-producible-by-a-probabilistic-pspace-machine-be-produced.
	\end{remark}
	\section*{Open question}
	Does Hypothesis \ref{hyp} hold?
	\section*{Acknowledgments}
	I would like to thank Nikolay Vereshchagin and Alexander Shen for useful discussions, advice and remarks.
	
	This work is supported by RFBR grant 16-01-00362 and supported in part by Young Russian Mathematics award and RaCAF ANR-15-CE40-0016-01 grant. The study has been funded by the Russian Academic Excellence Project `5-100'.

	\section*{Appendix}
	\subsection*{Symmetry of Information}
	Define  $\CD ^m(A,B)$ as the minimal length of a program that inputs a pair of strings $(a,b)$ and outputs a pair of boolean values $(a \in A, b \in B)$ using space at most $m$ for every input.
	
	\begin{lemma}[Symmetry of information] 
		\label{sim_inf}
		Assume $A, B \subseteq \{0,1\}^n$. Then
		$$ \textup{(a) } \forall m \text{ } \CD ^p(A, B)
		\le \CD ^m(A) + \CD ^m(B \cnd A)  + O(\log (\CD ^m(A,B)+m + n))$$ for $p = m + \poly(n + \CD ^m(A,B))$. 
		$$ \textup{(b) } \forall m \text{ } 
		\CD ^p(A) + \CD ^p(B \cnd A) \le \CD ^m(A, B) + O(\log (\CD ^m(A,B)+m + n) )$$ for $p = 2m + \poly(n + \CD ^m(A,B))$.
	\end{lemma}
	\begin{proof}[of Lemma~\ref{sim_inf} \textup{(a)}]
		The proof is similar to the proof of Theorem \ref{def} \textup{(a)}.
	\end{proof}
	\begin{proof}[of Lemma~\ref{sim_inf} \textup{(b)}]
		Let $k:= \CD ^m(A, B)$.
		Denote by $\mathcal{D}$ the family of sets $(U,V)$ such that  $\CD ^m(U,V) \le k$ and $U,V \subseteq \{0,1\}^n$. It is clear that  $|\mathcal{D}| < 2^{k+1}$. Denote by
		$\mathcal{D}_{A}$ the pairs of $\mathcal{D}$ whose the first element is equal to $A$.
		Let $t$ satisfy the inequalities $2^t \le |\mathcal{D}_{A}| < 2^{t+1}$.
		
		Let us prove that
		\begin{itemize}
			\item 
			$\CD ^p(B \cnd A)$ does not exceed $t$ significantly;
			\item
			$\CD ^p(A)$ does not exceed $k - t$ significantly.
		\end{itemize}
		Here $p=m + O(n)$.
		
		We start with the first statement. There exists a program that enumerates all sets from $\mathcal{D}_{A}$  using $A$ as an oracle and that works in space $2m + O(n)$. Indeed, such enumeration can be done in the following way: enumerate all programs of length $k$ and verify the following condition for every pair of $n$-bit strings. First, a program uses at most $m$ space on this input. Second, if a second  $n$-bit string belongs to $A$ then the program outputs $1$, and $0$ otherwise. Since some program loops we need aditional $m + O(n)$ space to take it into account. 
		
		Append to this  program the ordinal number of a program that distinguishes $(A,B)$. This number is not greater than $t+1$. Therefore we have $\CD ^p(B \cnd A) \le t + O(\log (\CD ^m(A,B) + m + n))$. 
		
		Now let us prove the second statement.
		Note that there exist at most $2^{k-t +1}$ sets $U$ such that $|\mathcal{D}_U| \ge 2^t$ (including $A$). Hence,  if we construct a program that enumerates all sets with such property
		(and  does not use much space) then we will win---the set $A$ can be described by the ordinal number of this enumeration.
		
		Let us construct such a program. It works as follows:
		
		enumerate all sets $U$ that are the first elements from $\mathcal{D}$, i.e. we enumerate programs that distinguish the corresponding sets (say, lexicographically). We go to the next step if the following properties holds. First, 
		$|\mathcal{D}_U| \ge 2^t$, and second: we did not meet set $U$ earlier (i.e. every program whose the lexicographical number is smaller does not distinguish $U$ or is not the first element from a set from $\mathcal{D}$). 
		
		This program works in  $2m + \poly(n + \CD ^m(A,B))$ space (that we want) and has length 
		$O(\log (\CD ^m(A)+ n +m))$. 
	\end{proof}
	
	\begin{proof}[of Lemma
		\ref{prob1}]
		Let us show that $\mathcal{B}$ satisfies property 
		\textup{$(1)^*$} with probability at most $2^{-n}$. 
		Since $\mathcal{B}$ satisfies property 
		\textup{(2)} with probability at most $\frac{1}{4}$ (see the proof of Lemma~\ref{prob}) it would be enough for us.
		
		For this let us show that every part is `bad' (i.e. has at least $(n + k)^2 + 1$ sets from $\mathcal{B}$) with probability at most $2^{-2n}$.  The probability of such event is equal to the probability of the following event: a binomial random variable with parameters $(2^k, 2^{-k}(n + 2)\ln 2)$ is greater than $(n + k)^2$. To get the needed upper bound for this probability is not difficult however the correspondent formulas are cumbersome. Take $w:=2^k$, $p:=2^{-k}(n + 2)\ln 2$ and $v:=(n + k)^2$. We need to estimate
		$$\sum_{i=v}^{w} {{w}\choose{i}} p^i(1-p)^{w-i} < w \cdot {{w}\choose{v}} p^v(1-p)^{w-v} < w \cdot {{w}\choose{v}} p^v < w \frac{(wp)^{v}}{v!}.$$
		The first inequality holds since $wp = (n+2) \ln 2 \le (n+k)^2 = v$. 
		Now note that $wp= (n+2) \ln 2 < 10 n$. So
		
		$$w \frac{(wp)^{v}}{v!} < 
		\frac{2^k (10n)^{(n+k)^2}}{((n+k)^2)!} 
		\ll 2^{-2n}.
		$$
	\end{proof}

\end{document}